\documentclass{article}[11pt,letter]
\pdfoutput=1

\usepackage[T1]{fontenc}
\usepackage[english]{babel}
\usepackage[cmex10]{amsmath}
\usepackage[utf8]{inputenc}
\usepackage{amssymb}
\usepackage{amsthm}
\usepackage{mathtools}
\usepackage{xspace}
\usepackage[symbol*]{footmisc}
\usepackage[hyperfootnotes=false,breaklinks,bookmarks=false]{hyperref}
\usepackage{algorithmic}
\usepackage{algorithm}
\usepackage{color}
\usepackage{authblk}
\usepackage{fullpage}

\usepackage[framemethod=tikz,skipbelow=3mm,skipabove=3mm,innertopmargin=0, nobreak=true]{mdframed}

\usepackage{cite}
\usepackage{microtype}

\usepackage{complexity}

\newtheorem{definition}{Definition}[section]
\newtheorem{lemma}[definition]{Lemma}
\newtheorem{theorem}[definition]{Theorem}
\newtheorem{proposition}[definition]{Proposition}

\newtheorem{corollary}[definition]{Corollary}

\newmdtheoremenv{defeo}{Definition}

\newcommand{\allperm}{\emph{All Permutations Supersequence}\xspace}

\newcommand{\notallperm}{\emph{Permutation Non-subsequence}\xspace} 
\newcommand{\csp}{\emph{Locally Constrained Permutation}\xspace}
\newcommand{\cps}{\emph{Prefix Increasing Permutation}\xspace}
\newcommand{\noncrossing}{\emph{Partially Non-crossing Perfect Matching in Bipartite Graphs}\xspace}

\newclass{\LCP}{LCP}
\newcommand{\csps}{\LCP\xspace}

\newcommand{\Poly}{\P}

\newclass{\threeSAT}{3SAT}
\renewcommand{\SAT}{\threeSAT\xspace}

\newcommand{\true}{\textsc{True}\xspace}
\newcommand{\false}{\textsc{False}\xspace}
\newcommand{\undecided}{\textsc{Undecided}\xspace}
\newcommand{\mem}[1]{\texttt{mem}(#1)}
\newcommand{\clause}[1]{\texttt{clause}(#1)}
\newcommand{\comp}[2]{\texttt{comp#1}(#2)}
\newcommand{\fcell}[1]{\texttt{f}(#1)}
\newcommand{\tcell}[1]{\texttt{t}(#1)}
\newcommand{\allowed}[1]{H_{#1}}
\newcommand{\etal}{{\it et~al.}\xspace}

\newcommand{\eg}{{\it e.g.}\xspace}
\newcommand{\ie}{{\it i.e.}\xspace}

\usepackage{setspace}   % just for \NOTE

\begin{document}

\author{Przemys\l{}aw Uznański}

\affil{Helsinki Institute for Information Technology HIIT,\\ Department of Computer Science, Aalto University, Finland}

\date{\today}

\title{All Permutations Supersequence is \coNP-complete}
\maketitle

\begin{abstract}
We prove that deciding whether a given input word contains as subsequence every possible permutation of integers $\{1,2,\ldots,n\}$ is \coNP-complete. The \coNP-completeness holds even when given the guarantee that the input word contains as subsequences all of length $n-1$ sequences over the same set of integers. We also show \NP-completeness of a related problem of \noncrossing, \ie to find a perfect matching in an ordered bipartite graph where edges of the matching incident to selected vertices (even only from one side) are non-crossing.

\end{abstract}

\section{Introduction and Preliminaries}

The question of deciding for two words whether one is a \emph{subsequence} of the other is one of the most basic problems in combinatorics. A folklore result states that a greedy solution is correct, and the problem is trivially in \Poly. However, if we consider the related questions of finding the shortest common supersequence or the longest common subsequence (\textsf{LCS}), both have been shown to be \NP-complete when allowed multiple input words, by Maier \cite{Maier:1978:CPS:322063.322075}, improved to binary alphabets for \textsf{LCS} by R{\"{a}}ih{\"{a}} and Ukkonen \cite{superseq-npc}.

The question of constructing the shortest word containing as subsequences all permutations (so called \emph{universal} words), was first posed by Knuth and attributed to Karp \cite{Chvatal:1972:SCR:891957}. More precisely, writing
  $f(n)$ for the length of such a shortest word where $n$ is the size of the alphabet, the question of determining the values of $f(n)$ was investigated (see \cite{oeisA062714} for known exactly values). Independently \cite{opac-b1004216,Adleman:1974:SPS:2625528.2625790,doi:10.1137/0130040,Koutas:1975:SSC:2625519.2625687,Mohanty:1980:SSC:2647973.2648192} provided an upper bound $f(n) \le n^2-2n+4$, while Newey \cite{opac-b1004216} proved that this is tight for $n\le7$. A stronger upper bound $f(n) \le \lceil n^2 - 7/3 n + 19/3 \rceil$ has been recently provided by Radomirovi\'{c} \cite{DBLP:journals/combinatorics/Radomirovic12}. Complementary, Kleitman and Kwiatkowski \cite{DBLP:journals/jct/KleitmanK76} have shown a lower bound of the form $f(n) \ge n^2 - C_{\varepsilon} n^{7/4+\varepsilon}$ for $\varepsilon > 0$. 

In this paper, we investigate the problem of deciding whether a given sequence is universal. The question on the hardness of this problem was first, to our knowledge, posed by Amarilli \cite{cstheory}. We prove that this problem is \coNP-complete, that is, a counterexample to the universality of any given word can be verified in polynomial time, but unless \Poly = \NP, no efficient algorithm exists to verify universality itself. Our result thus proves a separation between the problems of testing universality and testing whether an input word contains every word of given length (not only ones using distinct characters), with former being \coNP-complete and the latter being in \Poly~by a simple greedy algorithm.\footnote{The idea of the algorithm is as follows: iteratively take a letter of the alphabet for which the earliest occurrence after the current position is as far to the right as possible.}

In order to prove the \coNP-hardness of \allperm (see Definition~\ref{allpermdef}), we introduce the intermediate problem of \csp (see Definition~\ref{permdef}), which captures the essential hardness of former, while itself being much easier to work with. This problem, where we ask to reconstruct a permutation given sets of available values for each given position, and a list of linear order constraints which every pair of consecutive positions has to satisfy, falls into a larger category of \NP-complete problems involving permutation reconstruction. Similar problems were considered, \eg permutation reconstruction from differences (De Biasi \cite{DBLP:journals/combinatorics/Biasi14}) and recognizing sum of two permutations (Yu \etal~\cite{permutationsums}).

The reduction proving hardness of \csp can be shown to provide a very restricted instances. This fact, coupled with interpretation of permutations as perfect matchings in bipartite graphs, immediately provides us with a \NP-hardness result for \noncrossing (see Definition~\ref{nonc}). The problem of finding maximal non-crossing matching in bipartite graph, has been proposed and extended by Widmayer and Wong~\cite{DBLP:journals/ipl/WidmayerW85}, and the problem itself reduces to longest increasing subsequence, which is solvable in polynomial time (Fredman~\cite{Fredman197529}). Our result shows that lifting some non-crossing restrictions increases the computational complexity of the problem.

\paragraph{Notation.} In this paper we will denote permutations using greek letter $\pi$. To ease the notation, we will use $\pi$ both for the function $\pi\colon \{1,2,\ldots,n\} \to \{1,2,\ldots,n\}$ and for the word $\pi = \pi_1 \pi_2 \ldots \pi_n$. The set of all permutations of set $\{1, \ldots, n\}$ will be denoted as $S_n$. Given word $w$, we will write $w^R$ to denote $w$ reversed. For two words, we will write $w \sqsubseteq v$ meaning that $w$ is a subsequence of $v$. Given a set $S$, a linear order on $S$ is any binary relation $\prec$ such that for any two distinct $x,y \in S$ exactly one of $x \prec y$ or $y \prec x$ holds, and additionally $x \not\prec x$.

Now we are ready to formally define the problem that was already mentioned previously:

\begin{defeo}{\textnormal{\allperm}:}
\label{allpermdef}
\ \\\\
\textbf{Input:} Integer $n>0$, word $T$ over alphabet $\{1,\ldots,n\}$.\\\\
\noindent\textbf{Question:} For every $\pi \in S_n$, does it hold that
  $\pi \sqsubseteq T$?
\end{defeo}

An example of a shortest supersequence of all permutations of the set $\{1,2,3\}$ is provided in the Figure~\ref{fig:superseq}, together with placement of all permutations as subsequences.

\begin{figure}[t!]
\centering\includegraphics[scale=0.9]{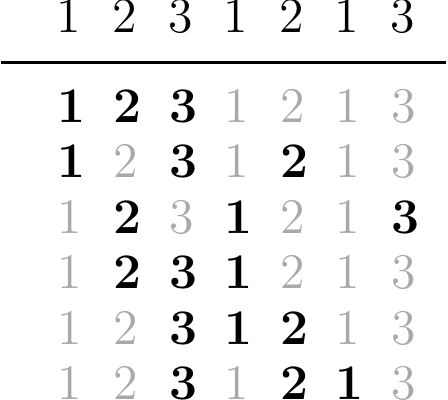}
\caption{The sequence $1,2,3,1,2,1,3$ satisfies \allperm for the set $\{1,2,3\}$. The length of $7$ is minimal.}
\label{fig:superseq}
\end{figure}

\newpage
\section{Locally Constrained Permutation}

First, we formally define the \csps problem.

\begin{defeo}{\textnormal{\csp}:}
\label{permdef}
\ \\\\
\textbf{Input:} Integer $n>0$, $n$ sets $\allowed{1},\ldots,\allowed{n} \subseteq \{1,\ldots,n\}$ and $n-1$ linear orders on $\{1,\ldots,n\}$: $\prec_1,\ldots,\prec_{n-1}$.\\\\
\noindent\textbf{Question:} Is there  $\pi \in S_n$ such that for each $1 \le i \le n$:
$\pi_i \in \allowed{i}$
and for each $1 \le i \le n-1$ we have
$\pi_i \prec_i \pi_{i+1}$?
\end{defeo}

To show the hardness of this problem, we will create an instance of \csps that encodes a given \SAT\ instance. %We will fix the value of $n$ later, first consecutively listing the positions in the output permutation to be set, together with possible values.
First, let us fix an instance of \SAT, which consists of: $m$ variables $v_1,\ldots,v_m$, and $d$ clauses of form $\ell_{i,1} \vee \ell_{i,2} \vee \ell_{i,3}$, where for each $1 \le i \le d, 1 \le j \le 3$: $\ell_{i,j} \in \{v_1, \ldots, v_m, \lnot v_1, \ldots ,\lnot v_m \} $.

We observe that for any given $i$ we don't need to fully specify the full order of $\prec_i$ on all of $\{1,\ldots,n\}$, as it is enough to specify it on $\allowed{i} \cup \allowed{i+1}$ only. We will also specify $\prec$ constraints not on every position, assuming it is possible for two consecutive positions to be under no constraint. Later we will show how to take care of this in \csps encoding.

The important property of \csps problem (and as well of any permutation reconstruction problem) is the fact that values can be used only once. Thus by assigning value to a position we are ``blocking'' this value from future use.

\paragraph{Literal encoding.}
We will encode each literal in a ``memory cell'' gadget. It consists of a separate position in the permutation with only two available choices, each of them corresponding to evaluating the underlying variable such that the literal evaluates to \true or \false, respectively. Thus, the information available can be ``carried'' over a long distance (to other occurrences of the same variable, or to positions evaluating truthfulness of the formula) by the fact that certain value is unblocked. However, this information is easily destroyed (one can think of it as read-once type of memory), thus we will need several working copies of the same memory cell. 

Let $p \le 3 d$ be the upper bound on the number of occurrences of a single variable in literals.
Thus, for each literal $\ell_{i,j}$, there will be consecutive positions $\mem{i,j}, \mem{i,j}+1, \ldots, \mem{i,j}+p$, together with distinct values $\fcell{i,j},\fcell{i,j}+1,\ldots,\fcell{i,j}+p$, $\tcell{i,j},\tcell{i,j}+1,\ldots,\tcell{i,j}+p$, such that $\allowed{\mem{i,j}} = \{\fcell{i,j},\tcell{i,j}\}, \ldots, \allowed{\mem{i,j}+p} = \{\fcell{i,j}+p,\tcell{i,j}+p\}$. To enforce proper value copying, we set $\prec_{\mem{i,j}}, \ldots, \prec_{\mem{i,j}+p-1}$ such that:
\begin{align*}
 \fcell{i,j} \prec_{\mem{i,j}} (\fcell{i,j}+1) & \prec_{\mem{i,j}} \tcell{i,j} \prec_{\mem{i,j}} (\tcell{i,j}+1),\\
 (\fcell{i,j}+1) \prec_{\mem{i,j}+1} (\fcell{i,j}+2)  &\prec_{\mem{i,j}+1} (\tcell{i,j}+1) \prec_{\mem{i,j}+1} (\tcell{i,j}+2),\\
&\ldots\\
(\fcell{i,j}+p-1) \prec_{\mem{i,j}+p-1} (\fcell{i,j}+p)  &\prec_{\mem{i,j}+p-1} (\tcell{i,j}+p-1) \prec_{\mem{i,j}+p-1} (\tcell{i,j}+p). 
\end{align*}

Observe, that there is a possibility for a one-sided error, that is assigning \false to $\mem{i,j}$  and \true to $\mem{i,j}+x$. However, those errors are not a problem for us, as they only occur when the literal is evaluated to \false, thus the value of this literal is irrelevant to the evaluation of this clause in satisfying assignment.

\paragraph{Clause evaluation.}
For each clause we add a single position gadget verifying that the clause evaluates to \true. Thus, for $i$-th clause, we have position $\clause{i}$ such that $\allowed{\clause{i}} = \{ \fcell{i,1}, \fcell{i,2}, \fcell{i,3} \}$. Thus, assigning value to position $\clause{i}$ will be possible iff at least one of literals it contains is evaluated to \true. 

\paragraph{Variable values consistency.}
To make sure that different occurrences of the same variable are assigned the same value, we use a literal equality gadget. Furthermore, we will say that a literal is \emph{positive} if it contains the simple variable, and is \emph{negative} if it contains the negated variable.

We iterate over all pairs of literals using the same variable. Let $\ell_{i,j}$ and $\ell_{i',j'}$ be respectively the $k$-th and $k'$-th occurrences of this variable. We will be using the $k'$-th copy of ``memory cell'' gadget of $\ell_{i,j}$ and $k$-th copy of $\ell_{i',j'}$, thus making sure that each copy is used at most once for comparison.

Additionally, for any such pair of literals, there are two unique positions $\comp{1}{i,j,i',j'}$ and $\comp{2}{i,j,i',j'}$ such that:
\begin{itemize}
\item $\allowed{\comp{1}{i,j,i',j'}} = \{\tcell{i,j}+k',\fcell{i',j'}+k\}$ and $\allowed{\comp{2}{i,j,i',j'}} = \{\fcell{i,j}+k',\tcell{i',j'}+k \}$ if both literals are positive or both are negative;
\item $\allowed{\comp{1}{i,j,i',j'}} = \{ \tcell{i,j}+k',\tcell{i',j'}+k\}$ and $\allowed{\comp{2}{i,j,i',j'}} = \{\fcell{i,j}+k',\fcell{i',j'}+k \}$ otherwise.
\end{itemize}

The satisfying assignment to positions $\comp{1}{i,j,i',j'}$ and $\comp{2}{i,j,i',j'}$ is thus possible (in first case) iff both $\tcell{i,j}+k'$ and $\tcell{i',j'}+k$ or both $\fcell{i,j}+k'$ and $\fcell{i',j'}+k$ values were unblocked (and the second case works by analogy).

\paragraph{Balancing the number of positions and values.}
The above construction introduces $n_1$ possible values and $n_2$ positions for some $n_2 \le n_1$. We introduce $n_2-n_1$ new positions $\pi_{n_1+1}, \ldots, \pi_{n_2}$ such that $\allowed{n_1+1}=\ldots=\allowed{n_2}=\{1,\ldots,n_2\}$, so that there exists a permutation $\pi \in S_{n_2}$ satisfying the larger instance iff there exists an injective function $\pi'\colon \{1,\ldots,n_1\} \to \{1, \ldots, n_2\}$ satisfying the original constraints.

\paragraph{Missing constraints.}
We also observe that we can choose not to impose any linear ordering restriction between any consecutive two positions $i$ and $i+1$ by inserting a dummy position between them. Thus, all positions $i+1,\ldots$ are moved one to the right, and a new dummy $i+1$ position is inserted. We also create a unique value $c_i$ such that $\allowed{i+1} = \{c_i\}$ and $c_i \not\in \allowed{j}$ for $j \not= i+1$, and and we construct the new linear
  orders $<_i$ and $<_{i+1}$ such that $c_i$ is the largest value with respect to $\prec_i$ and the smallest value with respect to $\prec_{i+1}$.

Thus all missing linear order constraints can be taken care of iteratively, adding one extra position and value for each.

\begin{lemma}
\label{lem:sat-csp}
An instance of \SAT is satisfiable iff a corresponding \csps instance is satisfiable.
\end{lemma}
\begin{proof}
For the "only if" direction, it is clear by construction that a solution to the \SAT instance can be used to construct a solution to the \csps instance.

To complete the proof, we need to show how one can reconstruct satisfying assignment to variables in \SAT from a permutation $\pi$ satisfying \csps instance.

Let us iterate over all the literals $\ell_{i,j}$. We will say that a literal is assigned \true (respectively \false) if for every $0\le k \le p$ the corresponding position $\mem{i,j}+k$ in $\pi$ holds value $\tcell{i,j}+k$ (respectively $\fcell{i,j}+k$), and that its value is \undecided otherwise.

Similarly, we will say that an occurrence of variable in literal is assigned value \true (respectively \false, \undecided) if the literal is positive and  is assigned the value of \true (respectively \false, \undecided) or the literal is negative and is assigned the value of \false (respectively \true, \undecided). We observe that for any two occurrences of the same variable, any configuration of values is allowed except one holding \true and another holding \false. However, as in the clause gadget we are using the first copy of any literal, any variable holding \undecided is not helping to evaluate the clause to \true (the corresponding $\mem{i,j}$ is assigned the value $\fcell{i,j}$). 

Consider assignment of values to variables of \SAT as follow: if there exists occurrence holding \true or \false, we assign \true or \false, respectively, and otherwise we assign any value arbitrarily. By previous reasoning, this is a satisfying assignment to the given \SAT instance.

\end{proof}

Clearly, \csps is in \NP\ and the reduction from \SAT to \csps is constructed in polynomial time. Thus, by Lemma~\ref{lem:sat-csp} we immediately get the following:

\begin{proposition}
\label{props}
\csps is \NP-complete.
\end{proposition}

Observe that in the reduction from \SAT to \csps we are explicitly using linear order constraints only for the memory cell gadget. Thus it is possible to define the values such that for any $1 \le i \le m, 1 \le j \le 3$:
$$\fcell{i,j} < (\fcell{i,j}+1) < \ldots < (\fcell{i,j}+p) < \tcell{i,j} < (\tcell{i,j}+1) < \ldots < (\tcell{i,j}+p)$$
and have every explicitly written $\prec_k$ be the ordering of integers. Additionally, we can make all the \mem{} positions smaller than any other ones, and possible values $\fcell{i,j}$ and $\tcell{i,j}$ are in the same order as for the corresponding positions, \ie $\mem{i,j} < \mem{i',j'}$ iff $\tcell{i,j} < \fcell{i',j}$. This way we make the reduction work so that the only linear order constraint used is monotonicity, it is used on some prefix of positions, and in this prefix even if there is no constraint, the possible values still satisfy monotonicity. Then, consider the following problem:

\begin{defeo}{\textnormal{\cps}:}
\ \\\\
\textbf{Input:} Integer $n>0$, $n$ sets $\allowed{1},\ldots,\allowed{n} \subseteq \{1,\ldots,n\}$ and integer $0 \le k \le n$.\\\\
\noindent\textbf{Question:} Is there $\pi \in S_n$ such that for each $1 \le i \le n$:
$\pi_i \in \allowed{i}$
and for each $1 \le j \le k$:
$\pi_{j} < \pi_{j+1}$
are satisfied?
\end{defeo}

It is worth noting, that both permutation reconstruction problems presented here have quite a natural interpretation in terms of matchings in bipartite graphs: each position $i$ in permutation corresponds to some vertex $a_i$, and each value $j$ corresponds to some vertex $b_j$, where we connect by edge $(a_i,b_j)$ iff $j \in \allowed{i}$.
Such a problem itself is naturally in \Poly, however additional linear order constrains we impose transform it into \NP-complete one.
 
 For example, \cps reduces to:
  
\begin{figure}[t!]
\centering\includegraphics[scale=1.2]{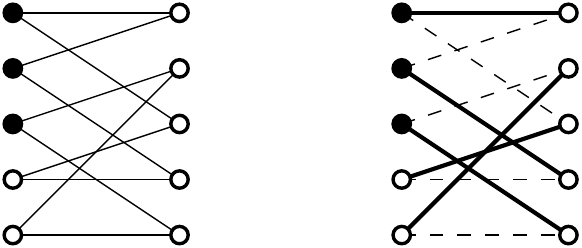}
\caption{An example of an instance of \noncrossing (on the left). We require that matching edges incident to blacked nodes are non-crossing.  A matching satisfying the constraints is on the right.}
\label{fig:noncrossing}
\end{figure}
  
\begin{defeo}{\textnormal{\noncrossing}:}
\label{nonc}
\ \\\\
\textbf{Input:} Ordered bipartite graph $G = (U,V,E)$ with orderings sets $U = (a_1,\ldots,a_n), V = (b_1,\ldots,b_n)$, and a subset $W \subseteq U$.
\\\\
\noindent\textbf{Question:} Is there a perfect matching $M \subseteq E$, such that the restricted matching $M' = M \cap (W \times V)$ is non-crossing, \ie if $(a_i,b_j),(a_k,b_l) \in M'$ then $i<j$ iff $k<l$?
\end{defeo}

\noindent An example of an instance to this problem is presented on a Figure~\ref{fig:noncrossing}. We have the following immediate corollary of Proposition~\ref{props}:

\begin{corollary}
\cps  and \noncrossing are both \NP-complete.
\end{corollary}

\section{All Permutations Supersequence}

Now we are ready to show hardness of \allperm problem. We will do it by analyzing the complementary problem:

\begin{defeo}{\textnormal{\notallperm}:}
\ \\\\
\textbf{Input:} Integer $n>0$ and word $T$ over alphabet $\{1,\ldots,n\}$.\\\\
\noindent\textbf{Question:} Is there $\pi \in S_n$ such that $\pi \not\sqsubseteq T$?
\end{defeo}

\begin{theorem}
\label{permnp}
\notallperm is \NP-complete.
\end{theorem}
\begin{proof}
\notallperm is clearly in \NP. Thus, it is enough  to construct a \Poly reduction from \csps to \notallperm.
Let us take an instance of \csps. Let us denote, given linear order $\prec$, by 
$\textsc{ORD}(\prec) = i_1 i_2 \ldots i_n$ a word build from the permutation defining the order, that is $i_1 \prec i_2 \prec \ldots \prec i_n$.
Similarly, given a set $H \subseteq \{1,\ldots,n\}$, let $\textsc{ENC}(H) = i_1 \ldots i_{\mid H \mid}$ be an arbitrary word containing every element of $H$. We also denote by $\overline{H} = \{1,2,\ldots,n\} \setminus H$.

Consider a word of the following form built from \csps instance:
$$W = \textsc{ENC}(\overline{H_1}) (\textsc{ORD}(\prec_1)^{R}) \textsc{ENC}(\overline{H_2}) (\textsc{ORD}(\prec_2)^{R}) \ldots (\textsc{ORD}(\prec_{n-1})^{R}) \textsc{ENC}(\overline{H_n}).$$
Clearly, $W$ can be constructed in \Poly.

We will show that for any given $\pi \in S_n$, $\pi$ is a feasible solution to the \csps instance iff $\pi$ is not a subword of $W$.
Observe that each $\textsc{ORD}(\prec_i)$ is a permutation of $\{1,\ldots,n\}$, so $W$ clearly contains as a subword any word of length $n-1$ (not necessarily a permutation).

To prove the \emph{if} part, observe that if $\pi$ is not a solution to the \csps instance, it must be for the following two reasons:
\begin{itemize}
\item For some $i$, $\pi_i \not\in H_i$. We have
  $\pi \sqsubseteq W$ for the following reason:
$\pi_j \in \textsc{ORD}(\prec_j)^R$ for $1\le j<i$, $\pi_i \in \textsc{ENC}(\overline{H_i})$ and $\pi_{j+1} \in \textsc{ORD}(\prec_j)^R$ for $i \le j < n$.
\item For some $i$, $\pi_i \not\prec_i \pi_{i+1}$. But then $\pi_i,\pi_{i+1}$ are exactly in this order in $\textsc{ORD}(\prec_i)^R$, hence, matching $\pi_j$ for $j < i$ and for $j > i$
  as in the previous case.
\end{itemize}

To prove the \emph{only if} part, let us take a $\pi \sqsubseteq W$. Let $i_1<i_2<\ldots<i_n$ be such that $\pi = W[i_1]W[i_2]\ldots W[i_n]$. At least one of the following conditions is fulfilled (as any subsequence contradicting both conditions at once can consist of positions from $\textsc{ORD}(\prec_i)^R$, one position per value of $i$, thus has length $n-1$ at most):
\begin{itemize}
\item There is $j$ such that position $i_j$ in $W$ is part of $\textsc{ENC}(\overline{H_j})$. But then $\pi_j \not\in H_j$, meaning that $\pi$ is not a solution to this \csps instance.
\item There is $j$ such that both $i_j$ and $i_{j+1}$ positions in $W$ are part of $\textsc{ORD}(\prec_j)^R$. But that implies $\pi_j \not\prec_j \pi_{j+1}$, thus $\pi$ is not a solution to this \csps instance. \qedhere
\end{itemize}
\end{proof}

As an immediate corollary of Theorem~\ref{permnp}, we have: 

\begin{corollary}
\allperm is \coNP-complete.
\end{corollary}

\section{Conclusion}
We proved the hardness of a problem of deciding whether a sequence is universal to every permutation with respect to having as subsequence. Somehow related to testing universality of a sequence with respect to certain combinatorial structures are following open questions on hardness of testing whether a sequence is an \emph{universal traversal sequence} (Aleliunas \etal \cite{4568017}) and whether a sequence is an  \emph{universal exploration sequence} (Koucký \cite{Koucky2002717}), that is whether a sequence defines a series of moves capable of exploring every (fixed size) connected graph.

\paragraph{Acknowledgements.} We are grateful both to Antoine Amarilli and to Jukka Suomela for providing with valuable suggestions and for fruitful discussions.
\bibliography{bib}
\end{document}